\documentclass[11pt,a4paper]{amsart}
\usepackage{amsmath,amsthm,amsfonts,amscd,amssymb,amstext}
\usepackage{booktabs}
\usepackage{esint}
\usepackage{graphicx,color}
\usepackage{tabularx}
\usepackage{titletoc} 
\usepackage{microtype}
\usepackage[colorlinks=true,linkcolor=black,citecolor=black,urlcolor=blue]{hyperref}
\usepackage{listings}
\usepackage{cite}
\usepackage{url}
\usepackage{alltt}

\theoremstyle{plain}
\newtheorem{theorem}{Theorem}

\newtheorem{lem}{Lemma}

\newtheorem*{conj}{Conjecture}

\theoremstyle{remark}
\newtheorem{rmk}{Remark}

\newcommand{\barr}{\begin{eqnarray}}
\newcommand{\earr}{\end{eqnarray}}
\newcommand{\be}{\begin{equation}}
\newcommand{\ee}{\end{equation}}

\newcommand{\de}{\mathrm{d}}

\let\baraccent=\= 
\renewcommand{\=}[1]{\stackrel{#1}{=}} 

\newcommand{\numberset}{\mathbb}
\newcommand{\N}{\numberset{N}}

\newcommand{\R}{\numberset{R}}

\newcommand{\avg}[1]{\left< #1 \right>} 

\newcommand{\Tr}{\mathrm{Tr}}
\newcommand{\T}{\mathcal{T}}

\newcommand{\kk}{\kappa}

\newcommand{\Cov}{\mathrm{Cov}}

\begin{document}

\title[Correlators for the Wigner-Smith time-delay matrix]{Correlators for the Wigner-Smith time-delay matrix of chaotic cavities} 

\author[Cunden]{Fabio Deelan Cunden}
\address[Fabio Deelan Cunden]{School of Mathematics, University of Bristol, University Walk, Bristol BS8 1TW, England}

\author[Mezzadri]{Francesco Mezzadri}
\address[Francesco Mezzadri]{School of Mathematics, University of Bristol, University Walk, Bristol BS8 1TW, England}

\author[Simm]{Nick Simm}
\address[Nick Simm]{Mathematics Institute, University of Warwick, Coventry CV4 7AL, United Kingdom}

\author[Vivo]{Pierpaolo Vivo}
\address[Pierpaolo Vivo]{King's College London, Department of Mathematics, Strand, London WC2R 2LS, United Kingdom}

\date{\today}

\begin{abstract}
We study the Wigner-Smith time-delay matrix $Q$ of a ballistic quantum dot supporting $N$ scattering channels.   We compute the $v$-point correlators of the power traces $\Tr Q^{\kk}$ for arbitrary $v\geq1$ at leading order for large  $N$ using techniques from the random matrix theory approach to quantum chromodynamics. We conjecture that the cumulants of the $\Tr Q^{\kk}$'s are integer-valued at leading order in $N$ and include a MATHEMATICA code that computes their generating functions recursively.

\end{abstract}

\maketitle

\section{Introduction} 
The proper delay times $\tau_1,\dots,\tau_N$ 
are the eigenvalues of the Wigner-Smith time-delay matrix $Q$, which is defined in term of the scattering matrix $S(E)$ by~\cite{Eis48,Smi60,Wig55}
\begin{equation}
  \label{eq:wdtm}
  Q = -i\hbar S^{\dagger}(E)\frac{\partial S(E)}{\partial E}.
\end{equation}
The Wigner-Smith matrix provides an overall description of the dynamical aspects of a scattering process 
in terms of the phase shift between a scattered wavefunction and a freely propagating one.
Ballistic quantum dots provide a physical realization of systems whose scattering matrix $S(E)$ is effectively modelled by a random unitary
matrix~\cite{BarMello94}, where $N$ is the total number of scattering channels.

The joint probability law of the rescaled inverse delay times $\lambda_i=(N\tau_i)^{-1}$ was derived
in~\cite{BroFraBee97_99} when the scattering matrix $S(E)$ belongs to the Dyson's circular ensembles~\cite{Dyson62} and was
generalised in~\cite{MBB14}  for systems whose symmetries belong to the classification introduced 
by Altland and Zirnbauer~\cite{Zir96,AZ97}. It is known in Random Matrix Theory (RMT) as the eigenvalue distribution of a Wishart-Laguerre matrix ensemble (Eq.~\eqref{eq:jpdf_rate}-\eqref{eq:V1} below). 
This ensemble belongs to the class of \emph{one-cut $\beta$-ensembles}, where $\beta$ is the Dyson index.
Without loss of generality, we will consider systems such that $\beta=1,2$ or $4$ according to the internal symmetries of the problem (see the reviews~\cite{Bee97,GuhMulWei1998,Texier15} for an extensive discussion). 

In this paper we study the power traces of the time-delay matrix $Q$, 
\be
\T_{\kk}=N^{\kk-1}\Tr Q^{\kk},\qquad \quad\kk\geq0,
\ee 
and we derive new formulae for their joint cumulants in the limit $N\to\infty$. Few basic facts on cumulants are collected  in Appendix~\ref{app:cumul}.  Let  $C_v(\T_{\kk_1},\dots,\T_{\kk_v})$ be the $v$-th cumulant of $(\T_{\kk_1},\dots,\T_{\kk_v})$. 
Take the limit
\begin{equation}
\lim_{N\to\infty}N^{2(v-1)}C_v(\T_{\kk_1},\dots,\T_{\kk_v})= \beta^{1-v}\alpha[\kk_1,\dots,\kk_v],\label{eq:first}
\end{equation}
with $\alpha[\kk_1,\dots,\kk_v]$ independent of $\beta$. 
Our main result is the derivation of the generating functions $F_{v,0}(z_1,\dots,z_v)$ of the limiting cumulants~\eqref{eq:first} for all $v\geq1$. 

Traces of the Wigner-Smith matrix for chaotic cavities have been investigated intensively in the past. Many results are available on the \emph{Wigner time-delay} $\T_1=\Tr Q$, both for finite and large $N$~\cite{FyoSom97,Sieber14,Simm11_12,Sommers01,Simm11_12,TM13}. A few results for higher-powers $\T_{\kk}=\Tr Q^{\kk}$ are also available. The covariances of $\T_1$ and $\T_2$ have been computed at leading order in $N$ in~\cite{Grabsch14}. Soon after one of the authors~\cite{Cunden14} considered the joint distribution of all the power traces $\T_{\kk}$ ($\kk\geq0$) and computed the averages and the full covariance structure at leading order in $N$. Recently, exact finite-$N$ formulae for averages of arbitrary products of power traces $\T_{\kk}$ ($\kk\geq1$) have been derived for $\beta=2$ \cite{Novaes15a}. 
The regime of non-ideal couplings with the external leads, of obvious importance in physical applications, has also been investigated~\cite{Savin01}, but is not considered in this paper. 

In deriving our results we follow ideas similar to those discussed in~\cite{CundenMV15}. Here, however, we prefer to use the convenient language and formalism of the so-called `loop calculus'~\cite{Migdal83} developed in the RMT approach to quantum chromodynamics. The functional derivative method discussed in~\cite{CundenMV15} can be concretely realized in this framework by means of the `loop insertion operator'. In addition, the derivation of our results is simplified by using a suitable `set of coordinates' first introduced by Ambj\o rn et al.\ in~\cite{Ambjorn93}. We are not aware of any general result on the leading order cumulants of random Wigner-Smith matrices~\eqref{eq:first} so far;
this is somewhat astonishing since the method of loop equations --- well established in a different area of physics --- can be applied directly.

The aim of this paper is thus threefold. First, we present a number of new results on the Wigner-Smith matrix of ballistic quantum dots.
We believe that the ideas discussed in this work may be useful in related problems;  using techniques similar to those used here it is possible, for instance, to obtain results for the moments of transmission matrices (see also recent applications of loop equation formalism in Gaussian and circular $\beta$-ensembles~\cite{ForresterWitte}). Hence, our second aim is to popularize the loop calculus in the random matrix approach of quantum transport in chaotic cavities.
Third, we provide strong evidences of an underlying combinatorial problem whose interpretation might be found through careful considerations of the semiclassical approaches to quantum chaos.

The plan of this paper is as follows. We start in Section~\ref{sec:gen} with some preliminaries on the one-cut $\beta$-ensembles. In Section~\ref{sec:scheme} the iterative scheme for the $v$-point Green functions is presented. Section~\ref{sec:main} contains the main result (Theorem~\ref{thm:1}) and its proof. In Section~\ref{sec:more} we discuss further results (Theorem~\ref{thm:2}) and a conjecture.
The paper is complemented by three appendices where we collect a few properties of cumulants (Appendix~\ref{app:cumul}), we present the derivation of some identities involved in the iterative scheme for the Green's functions (Appendix~\ref{app:calc}), and we include a MATHEMATICA code to implement in a systematic way the generating functions (Appendix~\ref{app:code}).

\section{Generalities about the one-cut ensembles}
\label{sec:gen}
Consider an ensemble of $N\times N$ random matrices
\be
P(\phi)=\frac{1}{Z} e^{-\beta N \Tr V(\phi)}\label{eq:ens}
\ee
where $\beta=1,2$ or $4$ if $\phi$ belongs to the set of real symmetric, complex hermitian or quaternion self-dual matrices, respectively. This ensemble is invariant under the adjoint action of the group $O(N), U(N)$ or $Sp(N)$ for $\beta=1,2$ or $4$, respectively, and hence it is usually referred to as an \emph{invariant ensemble}. We assume that the real \emph{potential} $V$ is independent of $N$ and $\beta$ and such that its derivative $V'$ (the force field) is a rational function. For our purposes it is convenient to consider the parametric family of potentials
\be
V(x)=t_0\log x+\sum_{k=1}^{\infty}\frac{t_k}{k}x^k,\qquad t_i\in\R,  \label{eq:V},
\ee
where it is understood that only finitely many $t_i$'s are nonzero. The nonzero couplings are assumed to be such that the partition function $Z=\int\de\phi e^{-\beta N \Tr V(\phi)}$ is finite. Of course, there are many admissible choices for the couplings (hence for $V$). One of the simplest is the quadratic potential $V(x)=x^2/4$ (that is $t_2=1/2$ and $t_i=0$ for $i\neq2$), corresponding to the celebrated Gaussian ensembles GOE, GUE and GSE. Here we anticipate that in our problem we will set the couplings $t_0=-t_1=-1/2$ and $t_i=0$ for $i\geq2$. Note that necessarily $t_0\leq0$ in~\eqref{eq:V} and, whenever $t_0\neq0$ it is implicitly assumed that the probability measure~\eqref{eq:ens} is supported on the set of positive definite matrices, $\phi\geq0$.

Averages of observables are defined in the usual way
\be
\avg{F(\phi)}=\frac{1}{Z}\int \de\phi F(\phi) P(\phi).\label{eq:avg}
\ee
Observables that depend uniquely on the spectrum $\{\lambda_i\}$ ($i=1,\dots,N$) of $\phi$ are particularly meaningful. In this case, \eqref{eq:avg} can be reduced to an integral over the joint probability distribution of the eigenvalues of $\phi$, which is given by 
\be
\de\mathbb{P}_{N,\beta}(\lambda_1,\dots,\lambda_N)= \frac{1}{\mathcal{Z}_{N,\beta}} e^{-\beta \left[ N\sum_{i}{V(\lambda_i)}-\sum_{i< j}\log|\lambda_i-\lambda_j|\right]}\prod_{k=1}^{N}\de\lambda_k. \label{eq:jpdf} 
\ee

For all $v\geq1$, let us define the \emph{$v$-point connected Green function} of $\phi$
\begin{align}
G_v(z_1,\dots,z_v)=C_v\left(\frac{1}{N}\Tr\frac{1}{z_1-\phi},\dots,\frac{1}{N}\Tr\frac{1}{z_v-\phi}\right). \label{eq:G_v}
\end{align}
Note that $G_v(z_1,\dots,z_v)$ is a nonrandom complex function of $v$ variables and is well-defined for $\Im z_i>0$. Of course, $G_v(z_1,\dots,z_v)$ is symmetric in its variables.
We denote the leading order term in $N$ of~\eqref{eq:G_v} by 
\be
G_{v,0}(z_1,\dots,z_v)=\lim_{N\to\infty}N^{2(v-1)}G_v(z_1,\dots,z_v).\label{eq:G_v0}
\ee
The indices `$v,0$' denote that $G_{v,0}$ is the generating function of the leading order in $N$ (the $g=0$ term in the `genus expansion' $1/N^g$) of the $v$-th order cumulants. The existence of a $1/N$ expansion has been proved in~\cite{Ercolani03} for polynomial potentials and $\beta=2$, and later for a general class of one-cut matrix models and every $\beta>0$ in~\cite{Guionnet13}. Note that, in general,  $G_{v,0}(z_1,\dots,z_v)$ is not continuous. For large $N$ the eigenvalues of $\phi$ concentrate on a subset of $\R$ that corresponds to the singularities of $G_{v,0}$. For instance, the discontinuity of the $1$-point Green function $G_{1,0}(z)=\lim_{N\to\infty}\avg{\frac{1}{N}\sum_{i=1}^N\frac{1}{z-\lambda_i}}$
on the real line is proportional to the \emph{limit density} of the eigenvalues of $\phi$.

We now assume that the potential $V$ is such  that the singularities of $G_{1,0}(z)$ consist of a finite single cut $[A,B]$, i.e.\ the limit density of the eigenvalues is concentrated on a single interval. Under this \emph{one-cut assumption} we have \cite{Ambjorn90} 
\be
G_{1,0}(z)=\oint_{\mathcal{C}}\frac{\de \omega}{2\pi i}\frac{V'(\omega)}{z-\omega}\sqrt{\frac{(z-A)(z-B)}{(\omega-A)(\omega-B)}}.\label{eg:G10}
\ee
Hereafter the integration contour $\mathcal{C}$ turns anticlockwise around the cut $[A, B]$ without enclosing the point $\omega=z$.
The real \emph{edges} of the cut $A$ and $B$ are determined by the normalization condition $G_{1,0}(z)\sim1/z$ as $z\to\infty$:
\be
\oint_{\mathcal{C}}\frac{\de \omega}{2\pi i}\frac{V'(\omega)}{\sqrt{(\omega-A)(\omega-B)}}=0,\quad \oint_{\mathcal{C}}\frac{\de \omega}{2\pi i}\frac{\omega V'(\omega)}{\sqrt{(\omega-A)(\omega-B)}}=1.\label{eq:edges}
\ee

\section{Iterative scheme for the correlators}
\label{sec:scheme}
The following ideas have been developed in the RMT approach to quantum chromodynamics and 2D quantum gravity (see~\cite{Ambjorn90,Ambjorn93,Akemann96,Itoi97,Chekhov06,Chekhov11,DiFrancesco95,Migdal83,Verbaarschot84}). We include them to make the paper self-contained. In the following, we compute the $v$-point correlators using the functional derivative method outlined below. The precise asymptotic meaning of these formal computations can be found in the mathematical papers~\cite{Ercolani03,Guionnet13}.
\subsection{Functional derivative method}
The whole family of Green's functions (at leading order in $N$) can be obtained from $G_{1,0}(z)$ using an iterative procedure involving functional derivatives. Indeed, we can write the variational formula,
\be
G_{v,0}(z_1,\dots,z_v)=-\frac{1}{\beta}\frac{\delta\,}{\delta V(z_v)}G_{v-1,0}(z_1,\dots,z_{v-1}),\label{eq:funct_der1}
\ee
where the functional derivative operator is defined as
\be
\frac{\delta}{\delta V(x)}=-\sum_{k=1}^{\infty}\frac{k}{x^{k+1}}\frac{d}{d t_k}.
\ee
Hence, for all $v>1$:
\be
G_{v,0}(z_1,\dots,z_v)=\left(-\frac{1}{\beta}\right)^{v-1}\frac{\delta \,}{\delta  V(z_v)}\cdots\frac{\delta \,}{\delta  V(z_2)}G_{1,0}(z_1).
\ee
\subsection{Total derivative formula}
The functional derivative method can be implemented by using the `loop-insertion operator' and a suitable set of variables. A natural choice would be to work with the couplings $t_i$'s that parametrize the potential $V(x)$ in~\eqref{eq:V}. However, it was realized in~\cite{Ambjorn93} that explicit calculations are easier by working in terms of the edges $A$ and $B$ and the coordinates $M_{\ell}$ and $J_{\ell}$ defined as
\begin{align}
M_{\ell}&=\oint_{\mathcal{C}}\frac{\de \omega}{2\pi i}\frac{V'(\omega)}{(\omega-A)^{\ell}\sqrt{(\omega-A)(\omega-B)}},\qquad \ell\geq0 \label{eq:Ml}\\
J_{\ell}&=\oint_{\mathcal{C}}\frac{\de \omega}{2\pi i}\frac{V'(\omega)}{(\omega-B)^{\ell}\sqrt{(\omega-A)(\omega-B)}},\qquad \ell\geq0. \label{eq:Jl}
\end{align}
Note that, from~\eqref{eq:edges}, $M_0=J_0=0$.

The \emph{loop-insertion operator}~\cite{Ambjorn90,Ambjorn93,Ercolani03} is defined as
\be
\frac{\partial\,}{\partial V(x)}=-\sum_{k=1}^{\infty}\frac{k}{x^{k+1}}\frac{\partial\,}{\partial t_k}.
\ee
The action of this operator has an interesting interpretation in the theory of symmetric forms on Riemann surfaces (see, for instance, \cite{Chekhov11}). 
Using $V'(x)=\sum_{k=0}^{\infty}t_kx^{k-1}$,  we have the formula~\cite{Ambjorn93,Eynard97}
\be
\frac{\partial V'(y)}{\partial V(x)}=-\frac{1}{(y-x)^2}.\label{eq:dV'dV}
\ee
The loop-insertion operator and the new set of variables provide a handy realization of the functional derivative in terms of a \emph{total derivative formula}:
\be
\frac{\delta\,}{\delta V(z)}=\frac{\partial\,}{\partial V(z)}+\frac{\delta A}{\delta V(z)}\frac{\partial\,}{\partial A}
+\frac{\delta B}{\delta V(z)}\frac{\partial\,}{\partial B}+\sum_{\ell\geq1}\left\{\frac{\delta M_{\ell}}{\delta V(z)}\frac{\partial\,}{\partial M_{\ell}}+\frac{\delta J_{\ell}}{\delta V(z)}\frac{\partial\,}{\partial J_{\ell}}\right\}.\label{eq:totald}
\ee
Therefore, in order to apply the functional derivative $\delta/\delta V(z)$, we need to compute the variation of $A$, $B$,  $M_{\ell}$ and $J_{\ell}$ with respect to the external potential $V(z)$:
\begin{align}
\frac{\delta A}{\delta V(z)}&=\frac{1}{M_1}\frac{1}{(z-A)\sqrt{(z-A)(z-B)}},\label{eq:dadV1}\\
\frac{\delta B}{\delta V(z)}&=\frac{1}{J_1}\frac{1}{(z-B)\sqrt{(z-A)(z-B)}},\label{eq:dbdV1}
\end{align}
and
\begin{align}
\frac{\delta M_{\ell}}{\delta V(z)}&=\left(\ell+\frac{1}{2}\right)\frac{\delta A}{\delta V(z)}\left(M_{\ell+1}-\frac{M_1}{(z-A)^{\ell}}\right)\label{eq:dMdV1}\\
&+\frac{1}{2}\frac{\delta B}{\delta V(z)}\left(\frac{J_1}{(B-A)^{\ell}}-\frac{J_1}{(z-A)^{\ell}}-\sum_{p=1}^{\ell}\frac{M_{p}}{(B-A)^{\ell-p+1}}\right),\nonumber\\
\frac{\delta J_{\ell}}{\delta V(z)}&=\left(\ell+\frac{1}{2}\right)\frac{\delta B}{\delta V(z)}\left(J_{\ell+1}-\frac{J_1}{(z-B)^{\ell}}\right)\label{eq:dJdV1}\\
&+\frac{1}{2}\frac{\delta A}{\delta V(z)}\left(\frac{M_1}{(A-B)^{\ell}}-\frac{M_1}{(z-B)^{\ell}}-\sum_{p=1}^{\ell}\frac{J_{p}}{(A-B)^{\ell-p+1}}\right).\nonumber
\end{align}
For completeness we present the derivation of these formulae in Appendix~\ref{app:calc}.
\subsection{Multi-point Green functions} 
\label{sub:scheme} The $1$-point Green function in the one-cut regime is given by~\eqref{eg:G10} supplemented by~\eqref{eq:edges}. A first functional differentiation provides the $2$-point correlator
\be
G_{2,0}(z_1,z_2)=\frac{1}{\beta}\frac{1}{(z_1-z_2)^2}\left[\frac{z_1z_2-(A+B)(z_1+z_2)/2+AB}{\sqrt{(z_1-A)(z_1-B)(z_2-A)(z_2-B)}}-1\right].\label{eq:G20_1}
\ee

Note that this result depends on the potential $V(x)$ only through the edges $A$ and $B$ of the cut (determined by~\eqref{eq:edges}). This is called \emph{macroscopic universality} of the $2$-point correlators \cite{BrezinZee93,Beenakker94}. $G_{2,0}(z_1,z_2)$  has no explicit dependence on $V(x)$. As a consequence of the functional derivative identity~\eqref{eq:funct_der1} the $v$-point Green functions ($v\geq2$) have no explicit dependence on the potential. However, the universality of the correlators decreases as $v$ increases. For instance, the $3$-point correlator is given by
\begin{align}
G_{3,0}(z_1,z_2,z_3)&=-\frac{1}{\beta}\frac{\delta\,}{\delta V(z_3)}G_{2,0}(z_1,z_2)\nonumber\\
&=-\frac{1}{\beta}\left[\frac{\delta A}{\delta V(z_3)}\frac{\partial\,}{\partial A}
+\frac{\delta B}{\delta V(z_3)}\frac{\partial\,}{\partial B}\right]G_{2,0}(z_1,z_2).
\end{align}
Note that $\frac{\delta A}{\delta V(z)}$ and $\frac{\delta B}{\delta V(z)}$ contain $M_1$ and $J_1$ (see~\eqref{eq:dadV1}-\eqref{eq:dbdV1}). Hence $G_{3,0}$ depends on $V$ through $A$, $B$, $M_1$ and $J_1$ (it is `less universal' than $G_{2,0}$). A further application of the functional derivative gives the $4$-point correlator 
\begin{align}
G_{4,0}(z_1,z_2,z_3,z_4)=-\frac{1}{\beta}&\Bigl[\frac{\delta A}{\delta V(z_4)}\frac{\partial\,}{\partial A}
+\frac{\delta B}{\delta V(z_4)}\frac{\partial\,}{\partial B}\nonumber\\
&+\frac{\delta M_1}{\delta V(z_4)}\frac{\partial\,}{\partial M_1}+\frac{\delta J_1}{\delta V(z_4)}\frac{\partial\,}{\partial J_1}\Bigr]G_{3,0}(z_1,z_2,z_3).
\end{align}
$G_{4,0}$ depends on  $A$, $B$, $M_1$, $J_1$, $M_2$, $J_2$. It is easy to see that in general, the $v$-point Green function ($v\geq2$) 
depends on the edges $A$ and $B$ and only on the first $v-2$ elements of the sequences $\{M_{\ell}\}$ and $\{J_{\ell}\}$. 

For general $v>2$ we can write 
\be
G_{v,0}(z_1,\dots,z_v)=-(1/\beta)\widehat{D}_v(z_v) G_{v-1,0}(z_1,\dots,z_{v-1})\qquad(v\geq3)\label{eq:recurrence1}
\ee
where $\widehat{D}_v(z)=\frac{\delta\,}{\delta V(z)}-\frac{\partial\,}{\partial V(z)}$ is the differential operator:
\begin{align}
\widehat{D}_v(z)&=\frac{\delta A}{\delta V(z)}\frac{\partial}{\partial A}+\frac{\delta B}{\delta V(z)}\frac{\partial}{\partial B}
+\sum_{\ell=1}^{v-3}\left\{\frac{\delta M_{\ell}}{\delta V(z)}\frac{\partial}{\partial M_{\ell}}
+\frac{\delta J_{\ell}}{\delta V(z)}\frac{\partial}{\partial J_{\ell}}\right\}.
\end{align}
The recurrence relation~\eqref{eq:recurrence1} supplemented by the initial datum $G_{2,0}(z_1,z_2)$ in~\eqref{eq:G20} provides the full family of $v$-point connected correlators at leading order in $N$ for the class of one-cut matrix ensembles. 

\begin{rmk} 
We emphasise that this route is independent of the specific model (the potential $V(x)$) and is valid within the class of one-cut ensembles. The leading order of the $v$-point Green function has the same structure determined by the one-cut property.  Equations~\eqref{eg:G10}-\eqref{eq:edges}-\eqref{eq:G20_1}-\eqref{eq:recurrence1} give the generic solution of the multi-point correlators $G_{v,0}$ as a function depending parametrically on $A$, $B$, $M_{\ell}$ and $J_{\ell}$. This general structure is applied to a specific model by inserting the model-dependent parameters $A$, $B$, $M_{\ell}$ and $J_{\ell}$ determined by~\eqref{eq:edges}-\eqref{eq:Ml}-\eqref{eq:Jl}.  (Note that the values of the edges $A$ and $B$ as well as $\{M_{\ell}\}$ and $\{J_{\ell}\}$ are functions of the couplings $t_i$ of the potential.) 

For $v\geq3$ the multi-point Green functions can be written as
\be
G_{v,0}(z_1,\dots,z_v)=\frac{1}{\beta^{v-1}}\cdot \frac{P_v(z_1,\dots,z_v)}{\prod_{i=1}^v\left[(z_i-A)(z_i-B)\right]^{v-3/2}},
\ee
where $P_v(z_1,\dots,z_v)$ is a symmetric polynomial of degree at most $2v-5$ in each variable. Note also that at leading order in $N$ the Dyson index $\beta$ plays no special role other than being a multiplicative parameter of the multi-point Green functions $G_{v,0}$. (See~\cite{Eynard97} for additional remarks.)
\end{rmk}

\section{Main Result}
\label{sec:main}
For all $v\geq1$ let us denote by $F_{v,0}(z_1,\dots,z_v)$ the generating function of the limiting cumulants $\beta^{1-v}\alpha[\kk_1,\dots,\kk_v]$ defined in~\eqref{eq:first}:
\be
F_{v,0}(z_1,\dots,z_v)=\beta^{1-v}\sum_{\kk_1=0}^{\infty}\cdots\sum_{\kk_v=0}^{\infty}\alpha[\kk_1,\dots,\kk_v]z_1^{\kk_1}\cdots z_v^{\kk_v}.
\ee 
The indices `$v,0$' denote that $F_{v,0}$ is the generating function of the of the $v$-th order cumulants at leading order in $N$ (the `genus' $g=0$ term) . Here we provide explicit expressions for $F_{v,0}(z_1,\dots,z_v)$.

The joint distribution of the inverse delay times $\lambda_i=(N\tau_i)^{-1}\geq0$ is~\cite{BroFraBee97_99} 
\be
\de\mathbb{P}(\lambda_1,\dots,\lambda_N)=\frac{1}{\mathcal{Z}_{N,\beta}}\prod_{i<j}|\lambda_i-\lambda_j|^{\beta}\prod_{k}\lambda_k^{\beta N/2}e^{-\beta N\lambda_k/2}\de\lambda_k.
\label{eq:jpdf_rate}
\ee
This distribution is of the form~\eqref{eq:jpdf} with external potential
\be
V(x)=\frac{1}{2}(x-\log x),\label{eq:V1}
\ee
which belongs to the general class~\eqref{eq:V} with the choice 
\be\nonumber
t_0=-1/2,\quad t_1=1/2\quad \text{and}\quad t_i=0 \quad\text{for}\quad i\geq2.
\ee 
Since $V(x)$ in~\eqref{eq:V1} is strictly convex, for large $N$ the density of inverse delay times concentrates on a single interval (the one-cut assumption is then satisfied). Hence, the joint law of the inverse delay times $\lambda_i$'s fits in the iterative scheme described in Section~\ref{sec:scheme}.

Before stating the result we need the following Lemma that specifies the actual values of the edges $A,B$ and the variables $M_{\ell},J_{\ell}$ for the Wigner-Smith inverse matrix. In the following we will use the lower-case symbols $a$, $b$, $m_1$, $j_1$, etc.\ to denote  these model-dependent values. Similarly, we denote the $v$-point Green function of the inverse delay times and its leading order in $N$ by 
\begin{align}
&g_{v}(z_1,\dots,z_v)=\frac{1}{N^v}C_v\left({\sum_{i_1=1}^{N}\frac{1}{z_1-\lambda_{i_1}},\cdots,\sum_{i_v=1}^{N}\frac{1}{z_v-\lambda_{i_v}}}\right),\label{eq:g_v}\\
&g_{v,0}(z_1,\dots,z_v)=\lim_{N\to\infty}N^{2(v-1)}g_{v}(z_1,\dots,z_v).\label{eq:g_v0}
\end{align}
In~\eqref{eq:g_v} the connected average refers to the joint distribution of the inverse delay times~\eqref{eq:jpdf_rate}.
\begin{lem}\label{lem:values} The leading order $1$-point Green function of the inverse delay times is
\be
g_{1,0}(z)=\frac{1}{2z}\left[z-1-\sqrt{\frac{(z-a)(z-b)}{ab}}\right]\label{eq:g_10}
\ee
where 
\be
a=3-\sqrt{8},\qquad b=3+\sqrt{8}.\label{eq:edges1}
\ee
Let us denote
\begin{align}
m_{\ell}&=\oint_{\mathcal{C}}\frac{\de \omega}{2\pi i}\frac{V'(\omega)}{(\omega-a)^{\ell}\sqrt{(\omega-a)(\omega-b)}},\qquad \ell\geq0 \label{eq:Ml1}\\
j_{\ell}&=\oint_{\mathcal{C}}\frac{\de \omega}{2\pi i}\frac{V'(\omega)}{(\omega-b)^{\ell}\sqrt{(\omega-a)(\omega-b)}},\qquad \ell\geq0, \label{eq:Jl1}
\end{align}
with $V$ given in~\eqref{eq:V1} and $a,b$ in~\eqref{eq:edges1}. Then
\begin{align}
m_0&=0,&j_0&=0,\\
\quad m_{\ell}&=\frac{(-1)^{\ell}}{2(3-\sqrt{8})^\ell},& j_{\ell}&=\frac{(-1)^{\ell}}{2(3+\sqrt{8})^\ell},\quad \text{for}\quad\ell\geq1. \label{eq:MlJl}
\end{align}
\end{lem}
\begin{rmk}
Equations~\eqref{eq:g_10}-\eqref{eq:edges1} mean that, as $N\to\infty$, the density of inverse delay times is the Mar\v{c}enko-Pastur distribution~\cite{MP} supported on the interval $[a,b]$:
\be
\rho(\lambda)=\frac{\sqrt{(\lambda-a)(b-\lambda)}}{2\pi \lambda}1_{a\leq \lambda\leq b}.\label{eq:rho}
\ee
\end{rmk}
We can now state our central result.
\begin{theorem}\label{thm:1}
For all integers $v\geq1$ we have
\begin{align}
F_{v,0}(z_1,\dots,z_v)
&=(-1)^v z_1\cdots z_v\, g_{v,0}(z_1,\dots,z_{v})+\delta_{1,v}(2-z_1) ,\label{eq:equality}
\end{align}
where $g_{v,0}(z_1,\dots,z_v)$ is the leading order in $N$ of the $v$-point Green function of the inverse delay times~\eqref{eq:g_v0}. For all $v\geq1$: 
\be
g_{v,0}(z_1,\dots,z_{v})={G}_{v,0}(z_1,\dots,z_{v})\Bigl|_{A=a,B=b,\{M_{\ell}\}=\{m_{\ell}\},\{J_{\ell}\}=\{j_{\ell}\}}.
\ee
The parametric functions ${G}_{v,0}$ are given by the one-cut iterative scheme described in Section~\ref{sub:scheme}:
\begin{align}
G_{1,0}&(z_1)=\frac{z_1-1-\sqrt{(z_1-A)(z_1-B)}}{\sqrt{4AB}z_1},\label{eq:G10}\\
G_{2,0}&(z_1,z_2)=\frac{1}{\beta}\frac{1}{(z_1-z_2)^2}\left[\frac{z_1z_2-(A+B)(z_1+z_2)/2+AB}{\prod_{i=1}^2\sqrt{(z_i-A)(z_i-B)}}-1\right],\label{eq:G20}\\
G_{v,0}&(z_1,\dots,z_v)=-(1/\beta)\widehat{D}_v(z_v)G_{v-1,0}(z_1,\dots,z_{v-1})\qquad(v\geq3).\label{eq:recurrence}
\end{align}
\end{theorem}
\begin{rmk}
The quantities $a$, $b$, $m_{\ell}$, $j_{\ell}$
are given explicitly in~\eqref{eq:edges1} and~\eqref{eq:MlJl}; note that  $[\delta a/\delta V(z)]$, $[\delta b/\delta V(z)]$, $\left[\delta m_{\ell}/\delta V(z)\right]$, $\left[\delta j_{\ell}/\delta V(z)\right]$ are also explicitly known from~\eqref{eq:dadV1}-\eqref{eq:dJdV1} above. 
\end{rmk}
For completeness we report the first four generating functions:
\begin{align}
&F_{1,0}(z_1)=\frac{1}{2}\left[3-z_1-\sqrt{z_1^2-6z_1+1}\right],\label{eq:F10}\\
&F_{2,0}(z_1,z_2)=\frac{1}{\beta}\frac{z_1z_2}{(z_1-z_2)^2}\left[\frac{z_1z_2-3(z_1+z_2)+1}{\sqrt{(z_1^2-6z_1+1)(z_2^2-6z_2+1)}}-1\right],\label{eq:F20}\\
&F_{3,0}(z_1,z_2,z_3)=\frac{16}{\beta^2}z_1z_2z_3\frac{[z_1 z_2 z_3 -(z_1+z_2+z_3)+6]}{\prod_{i=1}^3(z_i^2-6z_i+1)^{3/2}},\label{eq:F30}\\
&F_{4,0}(z_1,z_2,z_3,z_4)=\frac{32}{\beta^3}\frac{e_4}{\prod_{i=1}^4(z_i^2-6z_i+1)^{5/2}}\cdot\label{eq:F40}\\
&\bigl[3 e_4^3-9 e_3 e_4^2-5 e_2 e_4^2+333 e_1 e_4^2-3723 e_4^2+30 e_2 e_3 e_4-264 e_1 e_3 e_4\nonumber\\
&+2142 e_3 e_4-71 e_2^2 e_4+738 e_1 e_2 e_4-4354 e_2 e_4-1788 e_1^2 e_4+20538 e_1 e_4\nonumber\\
&-58863 e_4-2 e_2 e_3^2+18 e_1 e_3^2-156 e_3^2+3 e_2^2 e_3-20 e_1 e_2 e_3+126 e_2 e_3\nonumber\\
&+48 e_1 e_3-297 e_3+e_2^3-15 e_1 e_2^2+85 e_2^2+58 e_1^2 e_2-654 e_1 e_2+1843 e_2\nonumber\\
&-18 e_1^3+216 e_1^2-675 e_1+159\bigr],\nonumber
\end{align}
where $e_{k}=\sum_{i_1<i_2<\cdots<i_k}z_{i_1}z_{i_2}\cdots z_{i_k}$ denote the elementary symmetric polynomials in the variables $z_1,\dots,z_v$.
The expressions of $F_{v,0}$'s for $v\geq5$ are too
lengthy to be reported here. The general form for $v\geq3$ can be written as
\be
F_{v,0}(z_1,\dots,z_v)=\frac{e_v}{\beta^{v-1}}\frac{\sum C_{p_1,\dots,p_v}e_1^{p_1}\cdots e_{v}^{p_v}}{\prod_{i=1}^v(z_i^2-6z_i+1)^{v-3/2}},\label{eq:Fv0_gen}
\ee
(the sum runs over nonnegative integers $p_1,\dots,p_v$ such that $p_1+p_2+\dots+p_v\leq 2v-5$) with real coefficients $C_{p_1,\dots,p_v}$ that we were not able to determine in closed form. 
Standard symbolic manipulation softwares can  easily generate the $F_{v,0}$'s from~\eqref{eq:recurrence} and compute the coefficients $\alpha[\kk_1,\dots,\kk_v]$. In Appendix~\ref{app:code} we provide a MATHEMATICA code to implement the generating function $F_{v,0}$ for arbitrary $v$.

The generating function $F_{1,0}(z_1)$ in~\eqref{eq:F10} of the leading order averages $C_1(\T_{\kk})=\avg{\T_{\kk}}$ has been computed in~\cite{BerkKuip10}. The generating function $F_{2,0}(z_1,z_2)$ in~\eqref{eq:F20} of the covariances $C_2(\T_{\kk_1},\T_{\kk_2})=\Cov(\T_{\kk_1},\T_{\kk_2})$ has been obtained recently in~\cite{Cunden14}. Results on the leading order of the cumulants $C_v(\T_{1},\dots\T_{1})$ of the  Wigner-Smith time-delay $\T_1$  have been recently obtained in~\cite{Simm11_12} for generic $v\geq1$.  Theorem~\ref{thm:1} provides the generating functions $F_{v,0}(z_1,\dots,z_v)$ of the leading order cumulants $C_v(\T_{\kk_1},\dots\T_{\kk_v})$ for all $v\geq1$ and generic $\kk_i\geq1$. 
A few values of $\lim_N N^{2(v-1)}C_v(\T_{\kk_1},\dots\T_{\kk_v})$ are reported in Table~\ref{tab:I}. Some of them can be compared, with agreement, with previous results~\cite{BerkKuip10,Cunden14,Simm11_12,Novaes15a}.

\begin{proof}[Proof of Theorem~\ref{thm:1}]
To prove the Theorem we first establish the identity~\eqref{eq:equality} relating the generating functions  $F_{v,0}(z_1,\dots,z_v)$ and the leading order multi-point Green functions $g_{v,0}(z_1,\dots,z_v)$.
For a family of \emph{linear statistics} $X_N^{(j)}=N^{-1}\sum_{i=1}^N f_j(\lambda_i)$ on the inverse delay times $\lambda_i$ (here $f_j(x)$ denotes a generic function) the identity 
\be
C_v(X_N^{(1)},\dots, X_N^{(v)})
=\oint\frac{\de z_1}{2\pi i}\cdots\oint\frac{\de z_v}{2\pi i}f_1(z_1)\cdots f_v(z_v)g_v(z_1,\dots,z_v)\label{eq:contour1}
\ee
follows by Cauchy's integral formula with $g_v$ defined in~\eqref{eq:g_v}. In~\eqref{eq:contour1}, the integration contours enclose anticlockwise the inverse delay times $\lambda_i$ and no singularities of $f_1,\dots,f_v$.
For large $N$ we have
\be\label{eq:formulaint}
\lim_{N\to\infty}N^{2(v-1)}C_v(X_N^{(1)},\dots, X_N^{(v)})=\oint_{\mathcal{C}}\frac{\de z_1}{2\pi i}\cdots\oint_{\mathcal{C}}\frac{\de z_v}{2\pi i}g_{v,0}(z_1,\dots,z_v)\prod_{j=1}^vf_j(z_j).
\ee
Let us specialize~\eqref{eq:formulaint} to the Wigner-Smith time-delay moments $\T_{\kk}=N^{-1}\sum_{i=1}^Nf_{\kk}(\lambda_i)$, with $f_{\kk}(x)=x^{-\kk}$:
\be
\lim_{N\to\infty}N^{2(v-1)}C_v(\T_{\kk_1},\dots,\T_{\kk_v})=\oint_{\mathcal{C}}\frac{\de z_1}{2\pi i}\cdots\oint_{\mathcal{C}}\frac{\de z_v}{2\pi i}\frac{g_{v,0}(z_1,\dots,z_v)}{z_1^{\kk_1}\cdots z_v^{\kk_v}}.\label{eq:formulaint2}
\ee
The integral can be evaluated as the sum of residues. The possible singular sets are $z_i=0$ and $z_i=\infty$. For $v\geq2$, the integrand has no residue at $z_i=\infty$ and hence the only contribution comes from the residue at $z_i=0$:
\be
\lim_{N\to\infty}N^{2(v-1)}C_v(\T_{\kk_1},\dots,\T_{\kk_v})=(-1)^v\oint\limits_{|z_1|=\epsilon}\frac{\de z_1}{2\pi i}\cdots\oint\limits_{|z_v|=\epsilon}\frac{\de z_v}{2\pi i}\frac{g_{v,0}(z_1,\dots,z_v)}{z_1^{\kk_1}\cdots z_v^{\kk_v}},\label{eq:formulaint3}
\ee
where $\epsilon$ is sufficiently small and the factor $(-1)^v$ takes into account the fact that $z_i=0$ is `outside' the contour $\mathcal{C}$. The right-hand side of~\eqref{eq:formulaint3} is a generic coefficients of the power series expansion of $z_1\cdots z_v\, g_{v,0}(z_1,\dots,z_v)$ at $z_1=\cdots=z_v=0$, and this proves the identity~\eqref{eq:equality} for $v\geq2$. A similar argument applies to $v=1$, where the additional $\delta_{v,1}$ term in $F_{1,0}$ is due to the fact that $g_{1,0}$ is not analytic at $z=0$ and $z=\infty$. 

As already discussed, the Green functions of the inverse delay times $\lambda_i$'s are amenable to the iterative scheme presented in Section~\ref{sec:scheme}. Hence, assuming Lemma~\ref{lem:values}, the proof is complete.
\end{proof}
The rest of this section is devoted to proving Lemma~\ref{lem:values}.
\begin{proof}[Proof of Lemma~\ref{lem:values}]
The lemma can be easily proved by using residue calculus. We first compute the edges $a$ and $b$. For the potential~\eqref{eq:V1} of the inverse delay times  $V'(x)=(x-1)/(2x)$,~\eqref{eq:edges} reads
\be\nonumber
\begin{cases}
\displaystyle\mathop{\mathrm{Res}}_{\omega=\infty}\left(\frac{(\omega-1)\de\omega}{2\omega\sqrt{(\omega-a)(\omega-b)}}\right)-\mathop{\mathrm{Res}}_{\omega=0}\left(\frac{(\omega-1)\de\omega}{2\omega\sqrt{(\omega-a)(\omega-b)}}\right)=0, \\
\displaystyle\mathop{\mathrm{Res}}_{\omega=\infty}\left(\frac{\omega(\omega-1)\de\omega}{2\omega\sqrt{(\omega-a)(\omega-b)}}\right)=-1.
\end{cases}
\ee
We find the condition for the edges
\be\nonumber
\begin{cases}
a+b=6, \\
ab=1,
\end{cases}
\ee
and hence~\eqref{eq:edges1}. The $1$-point Green function is 
\begin{align}
G_{1,0}&(z)=-\frac{\sqrt{(z-a)(z-b)}}{2}\times\nonumber\\
&\left[\mathop{\mathrm{Res}}_{\omega=z}\left(\frac{(\omega-1)\de\omega}{\omega(z-\omega)\sqrt{(\omega-a)(\omega-b)}}\right)+\mathop{\mathrm{Res}}_{\omega=0}\left(\frac{(\omega-1)\de\omega}{\omega(z-\omega)\sqrt{(\omega-a)(\omega-b)}}\right)\right]\nonumber\\
&=\frac{z-1}{2z}-\frac{\sqrt{(z-a)(z-b)}}{\sqrt{4ab}z},\nonumber
\end{align}
and hence~\eqref{eq:g_10}.

 We now compute $\{m_{\ell}\}$ and $\{j_{\ell}\}$. For $\ell\geq1$ we find
\begin{align}\nonumber
m_{\ell}&=-\frac{1}{2}\mathop{\mathrm{Res}}_{\omega=0}\left(\frac{\de\omega}{\omega(\omega-a)^{\ell}\sqrt{(\omega-a)(\omega-b)}}\right)=-\frac{1}{2}(-1)^{\ell+1}\frac{1}{\sqrt{ab}}\frac{1}{a^\ell}, \\ \nonumber
j_{\ell}&=-\frac{1}{2}\mathop{\mathrm{Res}}_{\omega=0}\left(\frac{\de\omega}{\omega(\omega-b)^{\ell}\sqrt{(\omega-a)(\omega-b)}}\right)=-\frac{1}{2}(-1)^{\ell+1}\frac{1}{\sqrt{ab}}\frac{1}{b^\ell}. 
\end{align}
Inserting the values of $a=3-\sqrt{8}$ and $b=3+\sqrt{8}$ we eventually find~\eqref{eq:MlJl}.
\end{proof}

\section{Further results on the  cumulants and a conjecture}
\label{sec:more}
In the previous section we have obtained the generating functions of all cumulants of power traces of  $Q$. One may be satisfied with the belief that this result tells us everything about the time-delay matrix at leading order in the number of scattering channels $N$. However, we notice that all the values of $\alpha[\kk_1,\dots,\kk_v]$ reported in Table~\ref{tab:I} are integer numbers. 

In this section we address the questions: \emph{are the $\alpha[\kk_1,\dots,\kk_v]$'s integers?} and if yes, \emph{what do they count?} We feel that finding an answer to these questions is both challenging and intriguing. In this section we investigate these issues further.

\begin{lem}\label{lem:Legendre} Let $n$ be a nonnegative integer. Then $(z^2-6z+1)^{-n/2}$ is analytic around $z=0$ and the coefficients of the series expansion are nonnegative integers.
\end{lem}
\begin{proof}
The thesis is a straightforward consequence of the classical identity
\be
\frac{1}{\sqrt{z^2-2tz+1}}=\sum_{\ell=0}^{\infty}P_{\ell}(t)z^{\ell},
\ee
where $P_\ell(t)$ is the Legendre polynomial of degree $\ell$: 
\be\nonumber
P_{\ell}(t)=\frac{1}{2^{\ell} \ell !}\frac{\de\,}{\de t}(1-t^2)^{\ell}=\frac{1}{2^{\ell}}\sum_{p=0}^{{\ell}}\binom{{\ell}}{p}^2(t-1)^{{\ell}-p}(t+1)^p.
\ee
For $t=3$ we write explicitly
\be
P_{\ell}(3)=\sum_{p=0}^{{\ell}}\binom{{\ell}}{p}^2 2^p\in\N.
\ee
Hence, for all nonnegative integers $n$:
\be
\left(\frac{1}{\sqrt{z^2-6z+1}}\right)^n=\sum_{\ell=0}^{\infty}\left(\sum_{\substack{ \ell_1+\cdots+\ell_n=\ell}}P_{\ell_1}(3)\cdots P_{\ell_n}(3)\right)z^{\ell}.
\ee
\end{proof}
\begin{theorem}\label{thm:2} Let $\alpha[\kk_1,\dots,\kk_v]$ be defined as in~\eqref{eq:first}. Then
\begin{itemize}
\item[(i)] $\alpha[\kk]\in\N$;
\item[(ii)] $\alpha[\kk_1,\kk_2]\in\N$;
\item[(iii)] $\alpha[\kk_1,\kk_2,\kk_3]\in\N$;
\item[(iv)] $\alpha[\kk_1,\dots,\kk_v]\in\N$ for all $v\geq1$ if $\kk_i=1$ for all $i=1,\ldots,v$. 
\end{itemize}
\end{theorem}
\begin{proof}

(i) The first assertion is well-known. (The numbers $\alpha[\kk]=1,2,6,22,\dots$ ($\kk\geq1$) enumerate some combinatorial objects, see~\cite{Gou-BeauVauq88}, and they are referred to as the sequence of the `large Schr\"oder numbers'~\cite{Schroder}.)

(ii) The generating function of $\alpha[\kk_1,\kk_2]$ is given by  (see~\eqref{eq:F20})
\be
\beta F_{v,0}(z_1,z_2)=\frac{z_1z_2}{(z_1-z_2)^2}\left[\frac{z_1z_2-3(z_1+z_2)+1}{\prod_{i=1}^2\sqrt{z_i^2-6z_i+1}}-1\right]\label{eq:gen_part}
\ee
and is analytic at $(z_1,z_2)=(0,0)$. 
By Lemma~\ref{lem:Legendre} the term in the square bracket is analytic about $(z_1,z_2)=(0,0)$ with integer series coefficients. 
Note that
\begin{align}\nonumber
\frac{z_1z_2}{(z_1-z_2)^2}=\frac{z_2}{z_1}\frac{1}{\left(1-\frac{z_2}{z_1}\right)^2}=\xi\frac{\partial\,}{\partial \xi}\frac{1}{1-\xi}\biggl|_{\xi=z_2/z_1}=\sum_{n\geq0}n\xi^n\biggl|_{\xi=z_2/z_1}.
\end{align}
The products of power series with
integer coefficients is again a power series with integer coefficients; therefore, the coefficients $\alpha[\kk_1,\kk_2]$ of the series expansion of~\eqref{eq:gen_part} are integers. (Since~\eqref{eq:gen_part} is analytic, all the coefficients of negative powers of the Cauchy product vanish.) It is possible to perform the Cauchy product in closed form and obtain the explicit expression
\be
\alpha[\kk_1,\kk_2]=\frac{4\kk_1\kk_2}{\kk_1+\kk_2}\left[P(\kk_1)Q(\kk_2)+Q(\kk_1)P(\kk_2)\right],\label{eq:ak1k2}
\ee
where
\begin{align}
P(\kk)=\sum_{p=0}^{\kk-1}\binom{\kk-1}{p}\binom{\kk+p}{p},\qquad Q(\kk)=\sum_{q=0}^{\kk-1}\binom{\kk}{q+1}\binom{\kk+q}{q}.\label{eq:PQ}
\end{align}
This expression shows that $\alpha[\kk_1,\kk_2]$ is positive.

(iii) The proof goes along the same line of reasoning of point (ii). Note that $\alpha[\kk_1,\kk_2,\kk_3]=0$ if at least one among $\kk_1,\kk_2,\kk_3$ is zero. Let us consider the nontrivial case $\kk_1,\kk_2,\kk_3\geq1$. 
From~\eqref{eq:F30} we write
$F_{3,0}(z_1,z_2,z_3)=(16/\beta^2)z_1z_2z_3 f(z_1,z_2,z_3)$ with
\be\nonumber
f(z_1,z_2,z_3)=\frac{[z_1 z_2 z_3 -(z_1+z_2+z_3)+6]}{\prod_{i=1}^3(z_i^2-6z_i+1)^{3/2}}.
\ee
Let us introduce the notation ($\kk\geq0$):
\be\nonumber
a_{\kk}=\sum_{\ell_1+\ell_2+\ell_3=\kk}P_{\ell_1}(3)P_{\ell_2}(3)P_{\ell_3}(3)\in\N.
\ee
Clearly, this sequence  is monotone increasing  $a_{\kk+1}\geq a_{\kk}$. After standard manipulations one finds $f(z_1,z_2,z_3)=\sum_{\kk_1,\kk_2,\kk_3\geq0}b_{\kk_1,\kk_2,\kk_3}z_1^{\kk_1}z_2^{\kk_2}z_3^{\kk_3}$ with
\begin{align}
&b_{\kk_1,\kk_2,\kk_3}= \,6a_{\kk_1}a_{\kk_2}a_{\kk_3}+a_{\kk_1-1}a_{\kk_2-1}a_{\kk_3-1}
\nonumber\\
&-(a_{\kk_1-1}a_{\kk_2}a_{\kk_3}+a_{\kk_1}a_{\kk_2-1}a_{\kk_3}+a_{\kk_1}a_{\kk_2}a_{\kk_3-1})\nonumber\\
&=3a_{\kk_1}a_{\kk_2}a_{\kk_3}+a_{\kk_1-1}a_{\kk_2-1}a_{\kk_3-1}\nonumber\\
&+a_{\kk_2}a_{\kk_3}(a_{\kk_1}-a_{\kk_1-1})+a_{\kk_1}a_{\kk_3}(a_{\kk_2}-a_{\kk_2-1})+a_{\kk_1}a_{\kk_2}(a_{\kk_3}-a_{\kk_3-1})\in\N .\nonumber
\end{align}
Hence $\alpha[\kk_1,\kk_2,\kk_3]=16b_{\kk_1-1,\kk_2-1,\kk_3-1}\in\N$.

(iv) In~\cite{Simm11_12} it has been proved that for $\beta=2$, the cumulants of $\Tr Q$ are integers at leading order in $N$. In formulae, for $\beta=2$
\be\nonumber
\lim_{N\to\infty}N^{2(v-1)}C_v(\underbrace{\T_1,\dots,\T_1}_{v\text{ times}})\in\N,
\ee
for all $v\geq1$. Therefore, from~\eqref{eq:first} the thesis (iv) follows.
\end{proof}
\begin{rmk}
It is worth remarking that given enough algebraic computation (\emph{e.g.} by using the algorithm of Appendix \ref{app:code}), it is possible to verify that $\alpha[\kk_1,\kk_2,\ldots,\kk_v]$ are integers for a given large $v$ by the same method.
\end{rmk}
\begin{rmk}
We found that $P({\kk})$ and $Q({\kk})$ in~\eqref{eq:ak1k2}-\eqref{eq:PQ} are listed as A047781 and A002002, respectively, in the OEIS~\cite{Sloane}. Both sequences appear in a old paper concerning certain integral functions~\cite{Rutledge36}. 
\end{rmk}

Theorem~\ref{thm:2} and the inspection of the first few values of $\alpha[\kk_1,\dots,\kk_v]$ for $v>4$ (see Table~\ref{tab:I}) suggest the following conjecture.
\begin{conj}\label{conj} For all $v\geq1$ and $\kk_i\geq0$
\be
\alpha[\kk_1,\dots,\kk_v]\in\N.
\ee
\end{conj}
The fact that these cumulants are integer-valued is quite unexpected. 
Available finite-$N$ formulae~\cite{Simm11_12,Novaes15a} are cast in a form which is not suitable for a straightforward asymptotic expansion in $1/N$, and seem to be of scarce help to solve this conjecture. In addition, it is not clear what (if any) enumeration problem is solved by the family of numbers $\alpha[\kk_1,\dots,\kk_v]$ (so far a combinatorial interpretation is available only for $\alpha[\kk_1]$, see~\cite{Novaes11}). Some hints may come from semiclassical approaches to quantum transport. In fact, RMT predictions in quantum transport are usually confirmed by independent calculations for $\beta=1$ and $\beta=2$ based on semiclassical considerations. Averages in RMT correspond to sums over pairs of classical trajectories with \emph{encounters} connecting the leads with the interior of the cavity~\cite{BerkKuip10,Kuipers07,Sieber14,Vallejos98}. In absence of time reversal symmetry ($\beta=2$), the set of classical trajectories can be partitioned according to topological properties. Each class of trajectories has a representative \emph{diagram}. When computing cumulants (connected averages) one only needs to consider connected diagrams. The contribution $c(\mathcal{D})$ of a diagram $\mathcal{D}$ (i.e. the sum over all trajectories in the class represented by $\mathcal{D}$) can be evaluated by counting incoming channels $I(\mathcal{D})$, links $L(\mathcal{D})$ and encounters $E(\mathcal{D})$ inside the cavity containing at most one end point~\cite{Sieber14}
\be
c(\mathcal{D})=\frac{N^{\# I(\mathcal{D})}}{N^{\# L(\mathcal{D})}}\prod_{e\in E(\mathcal{D})}(-N)^{1-\#\text{end points in } e}.\label{eq:diagrams}
\ee 
Assuming the equivalence between RMT and semiclassical methods, the cumulants of the Wigner-Smith matrix $C_{v}(\T_{\kk_1},\dots,\T_{\kk_v})$ for $\beta=2$ are given by sums over $v$ sets of paired trajectories whose representative diagrams are connected.  These sums can be converted into sums over diagrams whose contribution is given by~\eqref{eq:diagrams}. 
The leading order in $N$ to $C_{v}(\T_{\kk_1},\dots,\T_{\kk_v})$ comes from diagrams $\mathcal{D}$ without loops that can be represented as \emph{rooted planar fat trees}~\cite{BerkKuip10,Sieber14,Novaes15a} (genus-$0$ maps). From~\eqref{eq:diagrams} it is clear that this leading order term is integer-valued. In view of the relation~\eqref{eq:first} between different Dyson indices, this argument shows the plausibility of our conjecture. 
A more precise mapping between RMT averages and semiclassical sum rules is likely to shed some light on the underlying combinatorics of $\alpha[\kk_1,\dots,\kk_v]$. This semiclassical argument, various finite-$N$ moments~\cite{Novaes15a} for $\beta=2$, and other evidences not reported here~\cite{Nick,Jack}, suggest that the conjecture could be extended to higher order corrections in $1/N$.
\begin{table}[h]
\caption{A few values of the limiting cumulants $\lim_{N\to\infty}N^{2(v-1)}C_v(\T_{\kk_1},\dots\T_{\kk_v})= \beta^{1-v}\alpha[\kk_1,\dots,\kk_v]$. Here $1\leq v\leq5$ and $1\leq \kk_i\leq 3$.}
\centering
    \begin{tabular}{ cr| cr | cr  }
    \hline\hline
    $({\kk_1})$ &  $\,\,\,\alpha[\kk_1]$ &   $({\kk_1},{\kk_2})$ &  $\alpha[\kk_1,\kk_2]$ &   $({\kk_1},{\kk_2},{\kk_3})$ &  $\alpha[\kk_1,\kk_2,\kk_3]$  \\ 
&&&&&\\
     $(1)$ & $1\quad$   &     $(1,1)$ & $4\quad$   &      $(1,1,1)$ & $96\quad$    \\ 
      $(2)$ & $2\quad$   &     $(1,2)$ & $24\quad$   &     $(1,1,2)$ & $848\quad$   \\ 
       $(3)$ & $6\quad$   &    $(1,3)$ & $132\quad$   &     $(1,1,3)$ & $6192\quad$   \\ 
       $$ & $$   &    $(2,2)$ & $160\quad$   &     $(1,2,2)$ & $7488\quad$  \\ 
       $$ & $$   &    $(2,3)$ & $936\quad$   &     $(1,2,3)$ & $54672\quad$   \\ 
       $$ & $$   &   $(3,3)$ & $5700\quad$   &       $(2,2,2)$ & $66112\quad$    \\ 
       $$ & $$   &   $$ & $$   &     $(1,3,3)$ & $399168\quad$    \\ 
       $$ & $$   &   $$ & $$   &      $(2,3,3)$ & $3524112\quad$    \\ 
       $$ & $$   &   $$ & $$   &      $(3,3,3)$ & $25729488\quad$   \\  
    \hline\hline
    \end{tabular}\vspace{5mm}
    \begin{tabular}{ cr | cr   }
    \hline\hline
        $({\kk_1},{\kk_2},{\kk_3},{\kk_4})$ &  $\alpha[\kk_1,\kk_2,\kk_3,\kk_4]$ &$({\kk_1},{\kk_2},{\kk_3},{\kk_4},{\kk_5})$ &  $\alpha[\kk_1,\kk_2,\kk_3,\kk_4,\kk_5]$  \\ 
&&&\\
    $(1,1,1,1)$ & $5088\quad$& $(1,1,1,1,1)$ & $437760$ \\ 
    $(1,1,1,2)$ & $54720\quad$ &$(1,1,1,1,2)$ & $5303808$ \\ 
    $(1,1,1,3)$ & $471552\quad$ &$(1,1,1,1,3)$ & $50969088$ \\ 
    $(1,1,2,2)$ & $569600\quad$ &$(1,1,1,2,2)$ & $61526016$ \\ 
    $(1,1,2,3)$ & $4794624\quad$ &$(1,1,1,2,3)$ & $572001408$ \\ 
    $(1,2,2,2)$ & $5792256\quad$ &$(1,1,2,2,2)$ & $690596352$ \\ 
    $(1,1,3,3)$ & $39648672\quad$ &$(1,1,1,3,3)$ & $5180744448$ \\ 
    $(1,2,2,3)$ & $47903616\quad$ &$(1,1,2,2,3)$ & $6255820800$ \\ 
    $(2,2,2,2)$ & $57876480\quad$ &$(1,2,2,2,2)$ & $7553912832$ \\ 
    $(1,2,3,3)$ & $390733632\quad$ &$(1,1,2,3,3)$ & $55488939264$ \\ 
    $(2,2,2,3)$ & $472119552\quad$ &$(1,2,2,2,3)$ & $67011268608$ \\ 
    $(1,3,3,3)$ & $3152001600\quad$ &$(2,2,2,2,2)$ & $80925462528$ \\ 
    $(2,2,3,3)$ & $3808797696\quad$ &$(1,1,3,3,3)$ & $483746017536$ \\ 
    $(2,3,3,3)$ & $30449851776\quad$ &$(1,2,2,3,3)$ & $584256789504$ \\ 
    $(3,3,3,3)$ & $241601800032\quad$ &$(1,2,3,3,3)$ & $5020755622272$ \\ 
    &&$(2,2,2,3,3)$ & $6064431920640$ \\ 
    &&$(1,3,3,3,3)$ & $42618105345024$ \\ 
    &&$(2,2,3,3,3)$ & $51481144857600$ \\ 
    && $(2,3,3,3,3)$ & $432425796811774$ \\ 
    &&$(3,3,3,3,3)$ & $3599012231119850$ \\ 
    \hline\hline
    \end{tabular}
\label{tab:I}
\end{table}
\section*{Acknowledgements}
FDC and FM acknowledge  support  from EPSRC Grant No.\ EP/L010305/1. FDC acknowledges partial support from the Italian National Group of Mathematical Physics (GNFM-INdAM). NS wishes to acknowledge support of a Leverhulme Trust
Early Career Fellowship (ECF-2014-309).  PV acknowledges the  stimulating  research  environment  provided  by  the  EPSRC  Centre  for  Doctoral Training  in  Cross-Disciplinary  Approaches  to  Non-Equilibrium  Systems  (CANES,EP/L015854/1).  FDC and PV thank Marcel Novaes for helpful correspondence.
FDC and FM would like to express their gratitude to Martin Sieber for
enlightening discussions on the semiclassical approaches to quantum transport.
FDC gratefully thanks Jack Kuipers for valuable correspondence and for sharing his unpublished results and Giuseppe Florio for technical support on the MATHEMATICA code. 

\appendix

\section{Cumulants}
\label{app:cumul}
Let us consider a collection of random variables $\xi_i$ ($i\geq1$) with finite moments of all order. The $v$-th cumulant of $\xi_1,\dots,\xi_v$ is defined  according to 
\be
C_v(\xi_1,\dots,\xi_v)=\sum_{\pi\in \mathcal{P}(\{1,\dots,v\})}(|\pi|-1)!(-1)^{|\pi|-1}\prod_{A\in\pi}\langle{\prod_{i\in A} \xi_i\rangle},\label{eq:cumul}
\ee
 where $\pi$ denote a generic partition of $\{1,\dots,v\}$ of size $|\pi|$. For instance, the lowest order cumulants are
\begin{align}
C_1(\xi_1)&=\avg{\xi_1},\qquad C_2(\xi_1,\xi_2)=\avg{\xi_1\xi_2}-\avg{\xi_1}\avg{\xi_2}\nonumber\\
C_3(\xi_1,\xi_2,\xi_3)&=\avg{\xi_1\xi_2\xi_3}\!-\!\avg{\xi_1\xi_2}\avg{\xi_3}\!-\!\avg{\xi_1\xi_3}\avg{\xi_2}\!-\!\avg{\xi_2\xi_3}\avg{\xi_1}+2\avg{\xi_1}\avg{\xi_2}\avg{\xi_3}
\nonumber
\end{align}
The variables $\xi_1,\dots,\xi_v$ in the definition~\eqref{eq:cumul}  do not need to be distinct. $C_v(\cdot)$, as a functional, is symmetric and multilinear, and $C_v(\xi_1,\dots,\xi_v)=\xi_1\delta_{v,1}$ whenever $\xi_i$ is a constant for some $i=1,\dots,v$.

\section{Identities in the total derivative formula}
\label{app:calc}
We report here the derivation of a few identities used in the iterative scheme of Section~\ref{sec:scheme}.
The variation of $M_{\ell}$ is
\begin{align}
\left[\frac{\delta M_{\ell}}{\delta V(z)}\right]=\frac{\partial M_{\ell}}{\partial V(z)}+\left[\frac{\delta A}{\delta V(z)}\right]\frac{\partial M_{\ell}}{\partial A}
+\left[\frac{\delta B}{\delta V(z)}\right]\frac{\partial M_{\ell}}{\partial B}.
\end{align}
Using~\eqref{eq:dV'dV}, the first term reads
\begin{align}
\frac{\partial  M_{\ell}}{\partial V(z)}&=\oint_{\mathcal{C}}\frac{\de \omega}{2\pi i}\frac{1}{(z-\omega)^2}\frac{1}{(\omega-A)^{\ell}\sqrt{(\omega-A)(\omega-B)}}\nonumber\\
&=-\mathop{\mathrm{Res}}_{\omega=z}\left(\frac{1}{(z-\omega)^2}\frac{1}{(\omega-A)^{\ell}\sqrt{(\omega-A)(\omega-B)}}\de\omega\right)\nonumber\\
&=-\lim_{\omega\to z}\frac{\de\,}{\de \omega}\frac{1}{(\omega-A)^{\ell}\sqrt{(\omega-A)(\omega-B)}}\nonumber\\
&=-\left(\ell+\frac{1}{2}\right)\left[\frac{\delta A}{\delta V(z)}\right]\frac{M_1}{(z-A)^{\ell}}-\frac{1}{2}\left[\frac{\delta B}{\delta V(z)}\right]\frac{J_1}{(z-B)^{\ell}}.
\end{align}
The derivative with respect to $A$ is
\begin{align}
\frac{\partial  M_{\ell}}{\partial A}&=\left(\ell+\frac{1}{2}\right)\oint_{\mathcal{C}}\frac{\de \omega}{2\pi i}\frac{V'(\omega)}{(\omega-A)^{\ell+1}\sqrt{(\omega-A)(\omega-B)}}=\left(\ell+\frac{1}{2}\right)M_{\ell+1}.
\end{align}
The last derivative requires a bit more work. Starting from
\begin{align}\nonumber
\frac{\partial  M_{\ell}}{\partial B}&=\frac{1}{2}\oint_{\mathcal{C}}\frac{\de \omega}{2\pi i}\frac{V'(\omega)}{(\omega-A)^{\ell}(\omega-B)\sqrt{(\omega-A)(\omega-B)}},
\end{align}
we observe that the following decomposition holds
\be\nonumber
f(\omega)=\frac{1}{(\omega-A)^{\ell}(\omega-B)}=\sum_{j=1}^{\ell}\frac{C_j}{(\omega-A)^j}+\frac{D}{(\omega-B)},
\ee
with $C_j=\mathop{\mathrm{Res}}\limits_{\omega=A}((\omega-A)^{j-1}f(\omega)\de\omega)=-1/(A-B)^{\ell-j+1}$ and $D=\mathop{\mathrm{Res}}\limits_{\omega=B}(f(\omega)\de\omega)=-1/(B-A)^{\ell}$. Hence one gets
\be
\frac{\partial  M_{\ell}}{\partial B}=-\frac{1}{2}\frac{J_1}{(B-A)^{\ell}}-\sum_{j=1}^{\ell}\frac{M_j}{(A-B)^{\ell-j+1}}.
\ee
The computation of $\left[\delta J_{\ell}/\delta V(z)\right]$ goes along identical lines. Eventually one finds \eqref{eq:dMdV1}-\eqref{eq:dJdV1}.

\section{Mathematica Code}
\label{app:code}
Below follows a MATHEMATICA code implementing 
the generating function $F_{v,0}(z_1,\dots,z_v)$ of cumulants of the Wigner-Smith matrix. 
The code provided
is tailored for the Wigner-Smith problem but can be easily modified for other problems.

The parameter \texttt{v} determines the highest order in the recursion: the code produces the generating functions $F_{1,0}, F_{2,0}, F_{3,0},\dots, F_{v,0}$ up to $v=\texttt{v}$. The code first implements the Green functions $G_{1,0},\dots,G_{v,0}$ using the iterative scheme outlined in Section~\ref{sec:scheme}. Then the set of the model parameters is assigned and the generating functions are produced.

The parameters $\beta, A,B,M_1,J_1$, etc.\ defined by~\eqref{eq:edges} and ~\eqref{eq:Ml}-\eqref{eq:Jl} correspond to \texttt{\textbackslash[Beta],A,B,M[1],J[1]}, etc.\ in the code.
The derivatives $\delta A/\delta V(z),\delta B/\delta V(z)$  involved in the total derivative formula~\eqref{eq:totald}
correspond to $\texttt{dAdV[z],dBdV[z]}$, while $\delta M_{\ell}/\delta V(z),\delta J_{\ell}/\delta V(z)$ for $\ell\leq v-3$ are collected in the arrays \texttt{dMdV[z]} and \texttt{dJdV[z]} of length $\texttt{v-3}$. (Recall that, in the computation of $F_{v,0}$ only the first $v-2$ terms of the $M_{\ell}$'s and $J_{\ell}$'s are required, and $M_0=J_0=0$.)

Note that in MATHEMATICA syntax any text in between the literals \texttt{(*} and \texttt{*)} is simply
a comment. The literal  $\backslash$ at the end of a line of input indicates that the expression on that line continues onto the next line.
The first part code  is given as follows.
\vspace*{2mm}
\hrule
\vspace*{2mm}
{\small
\begin{alltt}   
 (* \emph{Fix the order of the recursion} *)
\end{alltt}
\begin{verbatim}
Clear["Global`*"]
v = 5; 
\end{verbatim} 
\begin{alltt}   
 (* \emph{Define the derivatives of A, B, M_l & J_l with respect to V(z)} *)
\end{alltt}
\begin{verbatim}

dAdV[z_] := Power[M[1] (z - A) Sqrt[z - A]Sqrt[z - B], -1];

dBdV[z_] := Power[J[1] (z - B) Sqrt[z - A]Sqrt[z - B], -1];

dMdV[z_] := Table[
    (l + 1/2) dAdV[z] (M[l + 1] - M[1]/(z - A)^l) + \
    (1/2) dBdV[z] (J[1]/(B - A)^l - J[1]/(z - A)^l - \
       Sum[M[p]/(B - A)^(l - p + 1), {p, 1, l}]), {l, 1, v - 3}];
       
dJdV[z_] := Table[
    (l + 1/2) dBdV[z] (J[l + 1] - J[1]/(z - B)^l) + \
    (1/2) dAdV[z] (M[1]/(A - B)^l - M[1]/(z - B)^l - \
       Sum[J[p]/(A - B)^(l - p + 1), {p, 1, l}]), {l, 1, v - 3}];
\end{verbatim} 
}
\vspace*{2mm}
\hrule
\vspace*{2mm}
 Once run, the next part of the code implements the Green functions $G_{2,0},\dots, G_{v,0}$ denoted by \texttt{G[1]},\dots,\texttt{G[v]} in the code. (Note that $G_{1,0}$ is model-dependent and will be included later.) 
The two-point Green function $G_{2,0}$ is defined explicitly. The recursion based on the functional derivative method  (see Section~\ref{sec:scheme}) is implemented by a $\texttt{For}$ cycle starting from $G_{3,0}$.
\vspace*{2mm}
\hrule
\vspace*{2mm}
{\small
\begin{alltt}    
 (* \emph{Define the two-point correlator} *)
\end{alltt}
\vspace{-.4cm}
 \begin{verbatim}   

G[2] = (1/\[Beta]) (1)/(z[1] - z[2])^2 \
       ((z[1] z[2] - (A + B) (z[1] + z[2])/2 + A B)/ \
       (Sqrt[z[1] - A] Sqrt[z[1] - B] Sqrt[z[2] - A] Sqrt[z[2] - B]));
       
\end{verbatim}
\begin{alltt}    
 (* \emph{Iterative scheme based on the functional derivatives:
        compute recursively the higher order correlators} *)
 \end{alltt}
\vspace{-.4cm}
 \begin{verbatim} 
For[i = 3, i <= v, i++,
 If[i > 3, 
  G[i] = -(1/\[Beta]) \
     (dAdV[z[i]] D[G[i - 1], A] + dBdV[z[i]] D[G[i - 1], B] + \
      Sum[dMdV[z[i]][[q]] D[G[i - 1], M[q]] + \
        dJdV[z[i]][[q]] D[G[i - 1], J[q]], {q, 1, i - 3}]),\
  G[i] = -(1/\[Beta]) \
     (dAdV[z[i]] D[G[i - 1], A] + dBdV[z[i]] D[G[i - 1], B])
   ]]
 \end{verbatim}
}
\vspace*{2mm}
\hrule
\vspace*{2mm}
The next part of the code specifies the model-dependent parameters $a$, $b$, $\{m_{\ell}\}$ and $\{j_{\ell}\}$ denoted by \texttt{a, b, m, j}, respectively, and the one-point function $G_{1,0}$ (\texttt{G[1]} in the code). They are hard-coded in the following four lines.
This is the only section of the code that depends on the details of the potential $V(x)$.
\vspace*{2mm}
\hrule
\vspace*{2mm}
{\small
\begin{alltt}
 (* \emph{Parameters and one-point correlator of the Wigner-Smith matrix;
     NB: this is the only model-dependent part of the code} *)
 \end{alltt}
\vspace{-.4cm}
\begin{verbatim}
a = 3 - Sqrt[8]; b = 3 + Sqrt[8];
m = Table[(1/2) Power[-1/a, i], {i, 1, v - 2}];
j = Table[(1/2) Power[-1/b, i], {i, 1, v - 2}];
G[1] = (z[1] - 1 - Sqrt[z[1] - a]Sqrt[z[1] - b])/(2 z[1]);
\end{verbatim}
}
\vspace*{2mm}
\hrule
\vspace*{2mm}
In the last part of the code the numerical values of \texttt{a,b,m,j} are assigned to the symbols \texttt{A,B,M1,J1}, etc., and the list of generating functions $F_{1,0},\dots,F_{v,0}$  is constructed according to identity~\eqref{eq:equality}.
\vspace*{2mm}
\hrule
\vspace*{2mm}
{\small
\begin{alltt}
 (* \emph{Set the parameters} *)
 \end{alltt}
\vspace{-.4cm}
\begin{verbatim}
A = a; B = b;
For[i = 1, i <= v - 3, i++, M[i] = m[[i]]; J[i] = j[[i]]]
\end{verbatim}
\begin{alltt}
 (* \emph{Compute the generating functions of cumulants} *)
  \end{alltt}
\vspace{-.4cm}
\begin{verbatim}
For[i = 1, i <= v, i++, \
 F[i] = Product[-z[k], {k, 1, i}] G[i] + KroneckerDelta[i,1](2-z) ]
\end{verbatim}
}
\vspace*{2mm}
\hrule

\end{document}